\documentclass[twocolumn,aps,pra,notitlepage,10pt]{revtex4-2}
\usepackage{amsfonts,amssymb,amsmath}            
\usepackage{lmodern}
\usepackage[]{graphics,graphicx}            
\usepackage{amsthm}
\usepackage{bbold,enumitem,mathtools}
\usepackage[caption=false]{subfig}
\usepackage{tikz}
\usetikzlibrary{decorations.pathmorphing,patterns,decorations.markings,matrix,quantikz}

\usepackage{hyperref}

\usepackage[english]{babel}

\makeatletter
\def\bbl@set@language#1{%
  \edef\languagename{%
    \ifnum\escapechar=\expandafter`\string#1\@empty
    \else\string#1\@empty\fi}%
  \@ifundefined{babel@language@alias@\languagename}{}{%
    \edef\languagename{\@nameuse{babel@language@alias@\languagename}}%
  }%
  \select@language{\languagename}%
  \expandafter\ifx\csname date\languagename\endcsname\relax\else
    \if@filesw
      \protected@write\@auxout{}{\string\select@language{\languagename}}%
      \bbl@for\bbl@tempa\BabelContentsFiles{%
        \addtocontents{\bbl@tempa}{\xstring\select@language{\languagename}}}%
      \bbl@usehooks{write}{}%
    \fi
  \fi}
\newcommand{\DeclareLanguageAlias}[2]{%
  \global\@namedef{babel@language@alias@#1}{#2}%
}
\makeatother

\DeclareLanguageAlias{en}{english}

\theoremstyle{definition}

\newtheorem{lemma}{Lemma}
\theoremstyle{definition}

\DeclareMathAlphabet\mathbfcal{OMS}{cmsy}{b}{n}

\newcommand{\identity}{\mathbb{1}}

\begin{document}
\title{Noise reducing encoding strategies for spin chains}
\author{Catherine Keele}
\author{Alastair Kay}
\affiliation{Royal Holloway, University of London, Egham, Surrey TW20 0EX, England, United Kingdom}

\begin{abstract}
We present an encoding technique that reduces the effects of noise on quantum spin systems whose operation is driven by Hamiltonian evolution. This technique is widely applicable, being most relevant to the scenarios where there are insufficient qubits to permit full scale error correction. Instead, our technique can be implemented over small numbers of qubits and still leads to noticeable improvements in the fidelity of operations. The encoding scheme is easy to implement, flexible with respect to choice of Hamiltonian, and close to optimal.
\end{abstract}

\maketitle
\section{Introduction}
Quantum computers hold the promise of being able to solve problems that are too complex for classical computers \cite{shor1997}; however, fault tolerant quantum computation remains firmly in the future. To harness the power of quantum computation on near-term devices, there is currently a focus on Noisy Intermediate-Scale Quantum (NISQ) devices, which operate in the presence of noise but without access to full error-correction due to a limited number of qubits \cite{preskill2018}. Even in this setting, quantum supremacy has recently been demonstrated \cite{boixo2018}. Given the noise, but a lack of error correction, we must discover alternative strategies for error mitigation. Here, we introduce a simple technique that can mitigate against the effects of noise, improving the quality of operation in devices whose dynamics are driven by Hamiltonian evolution, using a small number of qubits relative to the number available.

The study of state transfer \cite{bose2003,christandl2004,bose2007,kay2010a} is a key test-bed for the development of ideas based upon Hamiltonian evolution. It provides a concrete task -- that of transferring an unknown quantum state from one qubit to another -- for us to study and demonstrates some of the key criteria for use in a NISQ device in that it is a useful elementary building block within more complex tasks and, in the Hamiltonian formulation, permits a factor of two speed enhancement compared to the circuit model \cite{yung2006} which,
while irrelevant to a computational scaling perspective, could be absolutely critical to achieving the maximum number of computational steps in a finite time before decoherence overwhelms the system.

The primary focus of early studies of state transfer was on perfect transfer \cite{christandl2004}. However, in an imperfect world, perfect transfer can never be achieved, and it is preferable to consider near-perfect transfer if, in trade, the system might be more tolerant of noise by, for example, achieving the transfer faster. Numerous methods have been considered for creating high quality transfer chains, from modifying some of the couplings in a chain \cite{apollaro2012} to encoding inputs and outputs \cite{osborne2004,haselgrove2005,keele2021}.

Nevertheless, these studies have generally focused on unitary evolution. Few studies have even quantified the effects of noise \cite{kay2010a,marais2013}, let alone attempted to directly improve the tolerance to noise. Those that have require unrealistic assumptions about the nature of the environment \cite{behzadi2018,burgarth2006a} or knowledge of when and where errors might happen \cite{marletto2012}.
Ultimately, these effects can be addressed by error correction \cite{burgarth2005,burgarth2005a,kay2017d,kay2016c} but we are interested in the regime for which error correction is not a realistic prospect because it requires too many qubits and too many operations.

Our approach is an encoding strategy that is lighter-touch than error correction. It generalises a method introduced by Haselgrove \cite{haselgrove2005}, in which an optimal encoding could be found for the task of quantum state transfer under unitary evolution. We further developed the technique in \cite{keele2021} for application to unitary evolution, but we now include the effects of noise. The method is broadly applicable to a wide range of Hamiltonians and noise types, and could readily be applied to tasks beyond that of state transfer. 

We show how two common types of noise -- amplitude damping and dephasing -- may be treated within our formalism, on an otherwise perfect system, and how fidelity in these types of system can be improved by use of encoding strategies. 
 
In section \ref{Setting}, we introduce the setting of our technique and the previous work by Haselgrove \cite{haselgrove2005} that we build upon here. Section \ref{Vectorisation} introduces our system and a measure of success.  Section \ref{Encoding} introduces our encoding scheme.  In Section \ref{Sec:SpecialCases} we incorporate the treatment of noise. We demonstrate the application of our encoding scheme to two example Hamiltonians -- one that implements perfect state transfer \cite{christandl2004,kay2010a}, and another specified in \cite{apollaro2012} which is the best chain that we know of in terms of the speed/fidelity trade-off. These are just expository and neither of these have been optimised for the scenario. We concentrate primarily on the single excitation subspace, but show how the results can be extended into the regime of multiple excitation subspaces in Section \ref{sec:multi}.
 
 \section{Setting}\label{Setting}

We consider a set of $N$ qubits, all prepared in the state $\ket{0}$. The task of state transfer requires the introduction of an unknown state $\ket{\psi}=\alpha\ket{0}+\beta\ket{1}$ in a given site, $a$, and for the evolution of the system Hamiltonian to cause this state to move to the output site, $b$. We assume that the Hamiltonian $H$ is excitation preserving, meaning $[H,\sum_{n=1}^{N}Z_{n}]=0$. As such, $\ket{\textbf{0}}:=\ket{0}^{\otimes N}$ is an eigenstate. (We will typically use $\ket{\textbf{0}}$ to denote the all-zero state of any number of qubits, where the number should be clear from context.) Thus, the only evolution that we have to focus on is that of the single excitation subspace, which is spanned by a basis $\ket{\textbf{n}} = \ket{0}^{\otimes(n-1)}\ket{1}\ket{0}^{\otimes(N-n)}$, within which the Hamiltonian is $H_1$. Without loss of generality, we take the input site to be labelled 1, and the output site $N$. Thus, perfect state transfer occurs in a system at time $t$ if $\left|e^{-iH_1t}_{N,1}\right|=1$. 

We will mostly focus on the single-excitation subspace. This makes the task computationally tractable, but is also motivated by \cite{keele2021}, in which it was shown that, for unitary evolution, the single excitation subspace is the optimal for encoding in across a broad parameter range. These encodings were capable of out-performing error correcting codes of the same size.

\subsection{Encoding}


Instead of controlling a single site at input and output, we assume control of two (small) sets of sites $\Lambda_{\text{in}}$ and $\Lambda_{\text{out}}$. We are able to prepare a single-excitation state $\ket{\psi}$ on the sites $\Lambda_{\text{in}}$,
$$
\ket{\Psi_{\text{in}}}=\alpha\ket{\textbf{0}}+\beta\ket{\psi}
$$
and receive it on the sites $\Lambda_{\text{out}}$. In the context of unitary evolution, the challenge of finding the optimal encoding and decoding strategy for a fixed $H_1$ and time $t$ was solved by Haselgrove \cite{haselgrove2005}. If
$$
P_{\text{in}}=\sum_{i\in\Lambda_{\text{in}}}\proj{\textbf{i}},\qquad P_{\text{out}}=\sum_{i\in\Lambda_{\text{out}}}\proj{\textbf{i}}
$$
are the projectors onto the input and output regions, then $P_{\text{in}}\ket{\psi}=\ket{\psi}$. We proceed by evaluating the $|\Lambda_{\text{out}}|\times|\Lambda_{\text{in}}|$ matrix
$$
\tilde U_t=P_{\text{out}}e^{-iH_{1} t}P_{\text{in}}.
$$
If we select an input state $\ket{\psi}$ to be a right singular vector of $\tilde U_t$, and $\ket{\phi}$ to be the corresponding left-singular vector, then by encoding in $\ket{\psi}$, and waiting a time $t$, the arriving state on the decoding region is $\ket{\phi}$ with a probability amplitude corresponding to the singular value. The maximum success probability is therefore just the square of the largest singular value of $\tilde U_{t}$.

We aim to generalise this encoding strategy to the context of noisy systems.

\section{Vectorisation}\label{Vectorisation}

In order to treat noise in a straightforward manner, we make use of a `vectorisation' procedure to the density matrix \cite{gilchrist2011} so that we have a density vector given as 
\begin{equation*}
    \ket{\rho} = \sum_{ij}\bra{ i}\rho\ket{ j}\ket{ ij}.
\end{equation*}
The advantage of this is that it converts noise super-operators into linear operators. As an example, for
$    \ket{\Psi_\text{in}} = \alpha\ket{\textbf{0}}+\beta\ket{\textbf{1}}$,
 the density matrix is $\rho = \proj{\Psi_\text{in}}$, and the density vector is
\begin{align*}
   \ket{\rho}    &= |\alpha|^2\ket{\textbf{00}}+\alpha^*\beta\ket{\textbf{10}}+\alpha\beta^*\ket{\textbf{01}}+|\beta|^2\ket{\textbf{11}} \\
   &=    \ket{\Psi_\text{in}\Psi^{*}_\text{in}}.
    \end{align*}
We need to know how operators on the density matrix $\rho$ manifest within this formalism.

\begin{lemma}\label{lem:vectorise}
Let $A$ and $B$ be two linear operators that act on a density matrix $\rho$. The vectorised form of $A\rho B$ is
\begin{equation*}
    \ket{ A\rho B} = A\otimes B^{T}\ket{\rho}.
\end{equation*}
\end{lemma}
\begin{proof}

Applying the definition for vectorisation, we have 
\begin{equation*}
    \ket{ A\rho B} = \sum_{kl}\bra{ k} A\rho B\ket{ l} \ket{ k,l}
\end{equation*}
Using the completeness relation yields
\begin{align*}
    \ket{A\rho B} &= \sum_{ijkl}\bra{k} A\proj{i}\rho\proj{j} B\ket{l}\ket{ k,l}     \\
    &=\sum_{ijkl}\ket{k,l}\bra{k}A\ket{i}\bra{l}B^T\ket{j}\bra{i}\rho\ket{j} \\
    &  = \left(\sum_{kl}\ket{ k,l}\bra{ k,l} A\otimes B^{T}\right)\sum_{ij}\bra{ i}\rho\ket{ j}\ket{ i,j}. \\
    &= A\otimes B^{T}\ket{\rho}.
\end{align*}
\end{proof}

As our system can be divided into subspaces, we use the notation $\ket{\rho_{00}}$, $\ket{\rho_{01}}$, $\ket{\rho_{10}}$, $\ket{\rho_{11}}$ to indicate the component of $\ket{\rho}$ on a given subspace. The $00$ subspace is a single element, $\ket{\textbf{00}}$, while the 11 subspace is spanned by basis states $\ket{\textbf{ij}}$. The 01 and 10 subspaces are the coherences between these two, and are spanned by $\ket{\textbf{0i}}$ and $\ket{\textbf{i0}}$ respectively.

\subsection{Noisy Evolution}

We describe noise using the Lindblad master equation 
\begin{equation*}
    \frac{d\rho}{dt}=-i[H,\rho]+\sum_{n=1}^{N}(L_{n}\rho L_{n}^{\dagger}-\frac{1}{2}L^{\dagger}_{n}L_{n}\rho-\frac{1}{2}\rho L_{n}^{\dagger}L_{n}),
\end{equation*}
where the $\{L_n\}$ specify the noise. Under our vectorisation technique, we write
\begin{equation*}
    \mathbfcal{Q} = \mathbfcal{H}+\sum_n L_{n}\otimes L_{n}^{*}-\frac{1}{2}L_{n}^{\dagger}L_{n}\otimes\identity-\frac{1}{2}\identity\otimes (L_{n}^{\dagger}L_{n})^*
\end{equation*}
where $\mathbfcal{H} = -iH\otimes\identity+i\identity\otimes H^{T}$. The Lindblad equation then becomes
\begin{equation*}
    \frac{d\ket{\rho}}{dt}=\mathbfcal{Q}\ket{\rho},
\end{equation*}
such that the evolution is given by
\begin{equation*}
    \ket{\rho(t)} = e^{\mathbfcal{Q}t}\ket{\rho(0)}.
\end{equation*}

\subsection{Trace}

We will also need to take the partial trace over a set of sites $\bar\Lambda$, leaving just the set of qubits $\Lambda$ remaining. We define
$$
T_{\Lambda}\ket{\rho}=\ket{\text{Tr}_{\bar \Lambda}(\rho)}.
$$
This is also a linear operator -- if $\{\ket{u_i}\}$ is an orthonormal basis over the qubits $\bar\Lambda$,
$$
\text{Tr}_{\bar \Lambda}(\rho)=\sum_i(\identity_{\Lambda}\otimes\bra{u_i})\rho(\identity_{\Lambda}\otimes\ket{u_i})
$$
such that, by Lemma \ref{lem:vectorise},
$$
T_{\Lambda}=\sum_i(\identity_{\Lambda}\otimes\bra{u_i})\otimes(\identity_{\Lambda}\otimes\bra{u_i^*}).
$$
Typically, one picks the standard basis for performing the trace, in which case $\ket{u_i^*}=\ket{u_i}$.

\subsection{Quality of Transfer}\label{sec:quality}

If our aim is to successfully transfer a state from one location to another, we must introduce a measure of success. For Hamiltonian evolution, the measure of success is the transfer fidelity \cite{christandl2004}, whose derivation we reproduce here, using the vectorised notation, before later (Sections \ref{Encoding} and \ref{Sec:SpecialCases}) expanding it to include the noisy evolution and encoding/decoding.
 
The density matrix after evolution is given as
 \begin{equation*}
     \ket{\rho^{\prime}} = U_{t}\otimes U_{t}^*\ket{\rho},
 \end{equation*}
where $U_{t} = e^{-iHt}$.
The fidelity of state transfer is given by
$$
F=(\alpha^*\bra{0}+\beta^*\bra{1})\otimes(\alpha\bra{0}+\beta\bra{1})T_N\ket{\rho^{\prime}}
$$
for a specific input state $\alpha\ket{0}+\beta\ket{1}$. However, to truly judge the efficacy of the protocol, one should average over all possible input states, by identifying $\alpha=\cos\frac{\theta}{2}$ and $\beta=\sin\frac{\theta}{2}$ such that the average fidelity of state transfer is
\begin{align*}
    \overline{F}&= \frac{1}{4\pi}\int_{0}^{2\pi}\int_{0}^{\pi} F \sin{\theta}d\theta d\phi, \\
    &=\frac{1}{6}\left(3+2\sqrt{F_{ex}}+F_{ex}\right)
\end{align*}
where, for an excitation preserving Hamiltonian,
$$
F_{\text{ex}}=\bra{11}T_Ne^{\mathbfcal{H}_{1}t}\ket{\textbf{11}}=|\bra{\textbf{N}}e^{-iH_1 t}\ket{\textbf{1}}|^2
$$
and $\mathbfcal{H}_{1} = -iH_{1}\otimes\identity+i\identity\otimes H^{T}_{1}$. Although we want our transfer fidelity to be as high as possible, we note that there is a natural threshold of $\frac{2}{3}$ \cite{bose2003}, which is the fidelity achieved by classically transferring a quantum state, that we need to beat.

\section{Encoding Strategy}\label{Encoding}

We now introduce our encoding scheme, where the initial state is encoded over the 0 and 1 excitation subspaces of a set of qubits $\Lambda_{\text{in}}$. Our aim is to find the ideal choice of state $\ket{\psi}$ for the initial encoding of $\ket{\Psi_{\text{in}}}=\alpha\ket{\textbf{0}}+\beta\ket{\psi}$, and a decoding unitary $U$ acting on the decoding region $\Lambda_{\text{out}}$ such that the state $\alpha\ket{0}+\beta\ket{1}$ is reproduced on a single site with the maximum fidelity. We are assuming here that $U$ is performing the decoding onto a single qubit in the decoding region. In \cite{keele2021} we took a different approach of encoding onto a separate ancilla qubit. However, our motivation here is that we will be encoding over as many qubits as we can, so that precludes the possibility of making additional qubits interact \footnote{Ultimately this choice makes no difference until we look at higher excitation subspaces, but is included to demonstrate a different assumption.}.

The initial state evolves through time according to
\begin{equation*}
\ket{\rho^{\prime}} = e^{\mathbfcal{Q}t}\left(\alpha\ket{\textbf{0}}+\beta\ket{\psi}\right)\left(\alpha^{*}\ket{\textbf{0}}+\beta^{*}\ket{\psi^{*}}\right),
\end{equation*}
which we will then decode. Upon decoding, we trace out all other qubits because they are irrelevant. Hence, the transfer fidelity is
\begin{equation*}
    F = (\alpha^{*}\bra{ 0}+\beta^{*}\bra{ 1^{*}})(\alpha\bra{ 0}+\beta\bra{ 1})T_NU\otimes U^*T_{\Lambda_\text{out}}\ket{\rho^{\prime}}
\end{equation*}
Let $\mathbfcal{R}=T_NU\otimes U^{*}T_{\Lambda_\text{out}}e^{\mathbfcal{Q}t}$. As in Sec.\ \ref{sec:quality}, we average over all possible input states (parameters $\alpha,\beta$) to give
\begin{multline*}
    \overline{F} = \frac{1}{6}\left(2\bra{00}\mathbfcal{R}\ket{\textbf{00}}+2\bra{11}\mathbfcal{R}\ket{\psi\psi^*}+\right.\\
    +\bra{00}\mathbfcal{R}\ket{\psi\psi^*}+\bra{11}\mathbfcal{R}\ket{\textbf{00}}+ \\
    \left.+\bra{01}\mathbfcal{R}\ket{\textbf{0}\psi^*}+\bra{10}\mathbfcal{R}\ket{\psi\textbf{0}}\right).
\end{multline*}
Since $\mathbfcal{R}\ket{\psi\psi^*}$ describes a one-qubit density matrix, which has trace 1, we have that $\bra{00}\mathbfcal{R}\ket{\psi\psi^*}+\bra{11}\mathbfcal{R}\ket{\psi\psi^*}=1$. We make one further assumption about the noise model -- that it's excitation non-increasing. This means that $e^{\mathbfcal{Q} t}\ket{\textbf{00}}=\ket{\textbf{00}}$. From this, we infer the decoding unitary should map
$$
U\ket{\textbf{0}}_{\Lambda_\text{out}}=\ket{\textbf{0}}_{\Lambda_\text{out}}.
$$
Consequently, $\bra{00}\mathbfcal{R}\ket{\textbf{00}}=1$ and $\bra{11}\mathbfcal{R}\ket{\textbf{00}}=0$. Hence, under this assumption,
\begin{multline}\label{FidelityEq}
    \overline{F} = \frac{1}{6}\left(3+\bra{11}\mathbfcal{R}\ket{\psi\psi^*}+\right.\\
    \left.+\bra{01}\mathbfcal{R}\ket{\textbf{0}\psi^*}+\bra{10}\mathbfcal{R}\ket{\psi\textbf{0}}\right).
\end{multline}
We remain free to choose $\ket{\psi}$ and the action of $U$ on the single-excitation subspace, to maximise $\overline{F}$. We start by considering the components $\bra{10}\mathbfcal{R}\ket{\psi\textbf{0}}$ and $\bra{11}\mathbfcal{R}\ket{\psi\psi^*}$ separately.

We start with $\bra{11}\mathbfcal{R}\ket{\psi\psi^*}$. How are we to pick $U$? Note that if there is an excitation on the decoding region, that can only result from an excitation in the input (as the noise cannot introduce excitations). As such, we definitely want to provide a $\ket{1}$ state on the output if possible. Consequently, we impose that
\begin{equation*}
    U\ket{\textbf{n}} = \ket{ N,\phi_{n}},\forall n\in\Lambda_{\text{out}},
\end{equation*}
i.e.\ a single excitation on the output spin, and some arbitrary state over the other qubits of the decoding region, subject to the $\{\ket{\phi_{n}}\}$ forming an orthonormal basis over $\Lambda_{\text{out}\setminus N}$.

If we explicitly write out the state after evolution, the only terms remaining in the $(1,1)$ subspace are
\begin{equation*}
    e^{\mathbfcal{Q} t}\ket{\psi\psi^{*}} \rightarrow \sum_{n,m=1}^{N}\gamma_{nm}\ket{\textbf{nm}}.
\end{equation*}
Thus,
\begin{align*}
    \bra{11}\mathbfcal{R}\ket{\psi\psi^{*}} &= \bra{11}T_NU\otimes U^*\sum_{n,m\in\Lambda_{\text{out}}}\gamma_{nm}\ket{\textbf{nm}} \\
    &=\bra{11}T_N\sum_{n,m\in\Lambda_{\text{out}}}\gamma_{nm}\ket{N,\phi_n}\ket{N,\phi_m^*}.
\end{align*}
Using the orthonormal basis $\{\ket{\phi_n}\}$ for the trace $T_N$ leaves
$$
\bra{11}\mathbfcal{R}\ket{\psi\psi^{*}}=\sum_{n\in\Lambda_{\text{out}}}\gamma_{nn}.
$$
This allows us to see that our choice of $U$ is irrelevant (beyond our earlier very natural assumptions), for the  $\bra{ 11} \mathbfcal{R}\ket{\psi\psi^{*}}$ component. We need now only to find the input state $\ket{\psi}$ to maximise $\sum\gamma_{nn}$.

Let 
\begin{equation*}
    R = \sum_{i,j\in\Lambda_{\text{in}}}\sum_{m\in\Lambda_{\text{out}}}\ket{ j}\bra{ i}\bra{ \textbf{mm}} e^{\mathbfcal{Q}t}\ket{ \textbf{ij}}
\end{equation*}
such that $\bra{11}\mathbfcal{R}\ket{\psi\psi^{*}}=\sum\gamma_{nn} = \bra{\psi} R\ket{ \psi}$.
It follows that we are able to maximise the $\bra{11}\mathbfcal{R}\ket{\psi\psi^{*}}$ component by selecting $\ket{\psi}$ to be the eigenvector of $R$ with the maximum eigenvalue.

We now continue on, to understand how to independently maximise the other component in Eq.\ (\ref{FidelityEq}). It is sufficient to maximise only $\bra{ 10} \mathbfcal{R}\ket{\psi\textbf{0} }$
 as this can always be made real by incorporating a phase on $U$, such that $\bra{ 10} \mathbfcal{R}\ket{\psi\textbf{0} } =\bra{ 01} \mathbfcal{R}\ket{\textbf{0}\psi^{*} }, $ and both are real.

For an excitation non-increasing $\mathbfcal{Q}$, after evolution we can parameterise the term
\begin{equation*}
 e^{\mathbfcal{Q} t}\ket{\psi \textbf{0}}=\gamma_{00}\ket{\textbf{00}} + \sum_{n=1}^N\gamma_{n0}\ket{\textbf{n0}}.
\end{equation*}
One can readily calculate
\begin{multline}
U\otimes U^*T_{\Lambda_\text{out}} e^{\mathbfcal{Q} t}\ket{\psi \textbf{0}}=\\\gamma_{00}\ket{\textbf{00}} +\sum_{n\in\Lambda_{\text{out}}}\gamma_{n0}\ket{N,\phi_n}\ket{\textbf{0}}.
\end{multline}
The application of the final trace yields
\begin{multline}
\mathbfcal{R}\ket{\psi \textbf{0}}=\gamma_{00}\ket{00} +\sum_{n\in\Lambda_{\text{out}}}\gamma_{n0}\braket{\textbf{0}}{\phi_n}\ket{10}.
\end{multline}
The required overlap
\begin{equation*}
  \bra{ 10} \mathbfcal{R}\ket{\psi\boldsymbol{0}}= \sum_{n\in\Lambda_{\text{out}}}\gamma_{n0} \bra{\textbf{0}}\phi_{n}\rangle=\sum_{n\in\Lambda_{\text{out}}}\gamma_{n0}\bra{\textbf{N}}U\ket{\textbf{n}}
\end{equation*}
is maximised by setting $U^\dagger\ket{\textbf{N}}$ parallel to $\sum_n\gamma_{n0}\ket{\textbf{n}}$. One component of each of these states is fixed,
\begin{equation*}
    \braket{\textbf{0} }{ \phi_{n}} = \frac{\gamma_{n}^{*}}{\sqrt{\sum_{m\in\Lambda}|\gamma_{m}|^{2}}}
\end{equation*}
and leaves them otherwise free.
This yields an optimal value of the component
\begin{equation*}
  \bra{ 10} \mathbfcal{\mathbfcal{R}}\ket{\psi \boldsymbol{0}} = \sqrt{\sum_{m\in\Lambda}|\gamma_{m}|^{2}},
\end{equation*}
which can alternatively be expressed as 
\begin{equation*}
    \left\| \left(P_{\text{out}}\otimes\bra{\textbf{0}}\right) e^{\mathbfcal{Q}t}\left(\ket{\psi}\otimes\ket{\textbf{0}} \right)\right\|.
\end{equation*}
In the case of unitary evolution, this recovers the result of Haselgrove \cite{haselgrove2005}.
Therefore our encoding $\ket{\psi}$ will be the eigenvector corresponding to the largest right singular vector of 
\begin{equation}\label{eq:FinalQ}
   S=\left(P_{\text{out}}\otimes\bra{ \textbf{0}}\right) e^{\mathbfcal{Q}t}\left(P_{\text{in}}\otimes\ket{ \textbf{0}}\right).
\end{equation}
Instead of finding the maximum right singular vector of $S$, we can alternatively find the maximum eigenvector of $S^{\dagger}S$.

Note that, in optimising the two components separately, the conditions on the choice of $U$ are mutually compatible, it is only the choice of encoding $\ket{\psi}$ that could potentially differ. The overall expression for the fidelity is thus
\begin{equation}\label{eqn:fbar}
\bar{F}=\frac12+\frac13\sqrt{\bra{\psi}S^\dagger S\ket{\psi}}+\frac16\bra{\psi}R\ket{\psi}.
\end{equation}
An exactly optimal choice of $\ket{\psi}$ in all circumstances is non-trivial, but we will see in Sec. \ref{Sec:SpecialCases} that there are instances where this can be solved exactly. Furthermore, an extremely good approximation can be made in many reasonable cases. Equally, our intended operating regime for these encodings is with small sizes of encoding/decoding region, for which exact optimisation is possible.

\section{Special Cases}\label{Sec:SpecialCases}

We will now study some special cases in which we can find the optimal $\ket{\psi}$ and hence evaluate $\bar F$.

\subsection{Unitary Evolution}

We now derive the optimal encoding strategy in the noise-free case, which coincides exactly with Haselgrove's strategy \cite{haselgrove2005}. We consider, first, the $R$ term.

\begin{equation}\label{eq:Runitary}
    R = \sum\ket{ j}\bra{ i}\bra{\textbf{mm}} e^{(-iH_1\otimes\identity+i\identity\otimes H_1^{T})t}\ket{\textbf{i,j}}.
\end{equation}
We can rearrange this to get
\begin{align*}
    R &=\sum_{i,j,m}\ket{ j}\bra{i}\bra{\textbf{j}} e^{iH_1t}\ket{\textbf{m}}\bra{\textbf{m}} e^{-iH_1t}\ket{\textbf{i}} \\
    &=\sum_{i,j}\ket{ j}\bra{i}\bra{\textbf{j}} e^{iH_1t}P_{\text{out}} e^{-iH_1t}\ket{\textbf{i}} \\
    &=S^\dagger S.
\end{align*}
In this case, the $R$ term does not further constrain our choice of $\ket{\psi}$, and we are free to pick it to optimise the $S$ term via Eq.\ (\ref{eq:FinalQ}), thereby reproducing the strategy of Haselgrove. 

\subsection{Large Decoding Region}

Consider the case where $\Lambda_{\text{out}}$ comprises every qubit in the system. This is clearly not a realistic scenario, but is nevertheless interesting. We restrict the noise model to being excitation preserving, with the maximally mixed state (of each excitation subspace) being the fixed points of the map. One such example is dephasing noise. The state
$$
\sum_{m=1}^N\ket{\textbf{mm}}
$$
is an eigenstate of $\mathbfcal{Q}$ because it's the maximally mixed state, and thus the fixed point of the map. Hence, $R=\identity$. Again, the choice of $\ket{\psi}$ is irrelevant to this term, and we need only consider the term arising from $S$.

\subsection{Amplitude Damping Noise}\label{sec:ampdamp}

One important type of noise that we can treat within this formalism is amplitude damping noise. This describes the relaxation of a system as it loses energy to the environment. The Lindblad operators are $L_{i} = \sqrt{\frac{\Gamma_{x}}{2}}(X_{i}+iY_{i})$, where $\Gamma_{x} = \frac{1}{T_{1}}$ is a measure of the strength of the noise, $T_{1}$ being the longitudinal paramagnetic relaxation time \cite{rost2020}.

Due to the structure of the amplitude damping noise $\mathbfcal{Q}_{A}$, we can find analytical solutions for evolution constrained within the zero and single excitation subspaces by considering the evolution of each subspace of the density vector separately. We first look at evolution of the $\ket{ 11}$ subspace which is given by
\begin{equation*}
    \frac{d\ket{\rho_{11}}}{dt} = (-2\Gamma_{x}\identity+\mathbfcal{H}_{1})\ket{\rho_{11}}.
\end{equation*}
This has a direct solution
$$
\ket{\rho_{11}(t)}=e^{-2\Gamma_x t}e^{\mathbfcal{H}_{1}t}\ket{\rho_{11}(0)}
$$

Applying the same logic to the other subspaces of $\ket{\rho}$ tells us that the components evolve as
\begin{align}
    \ket{\rho_{00}}&\rightarrow\ket{\rho_{00}}, \label{eq:evolvep} \\
    \ket{\rho_{01}}&\rightarrow e^{-\Gamma_{x}t}e^{iH_{1}t}\ket{\rho_{01}},\\
    \ket{\rho_{10}}&\rightarrow e^{-\Gamma_{x}t}e^{-iH_{1}t}\ket{\rho_{10}},\\
    \ket{\rho_{11}}&\rightarrow e^{-2\Gamma_{x}t} e^{\mathbfcal{H}_{1}t}\ket{\rho_{11}}+(1-e^{-2\Gamma_{x}t})\ket{\rho_{00}}.\label{eq:evolvepend}
\end{align}
This result gives us an exact solution for the evolution of the density vector when subject to amplitude damping noise and shows that the effect of this noise is just to add the deterioration term $\Gamma_{x}$. We can see instantly that solutions can be taken directly from the noise-free case and our only opportunity to minimise noise is to make transfer as fast as possible.

\begin{figure*}
\subfloat[Uniformly coupled Hamiltonian]{\includegraphics[width=0.45\textwidth]{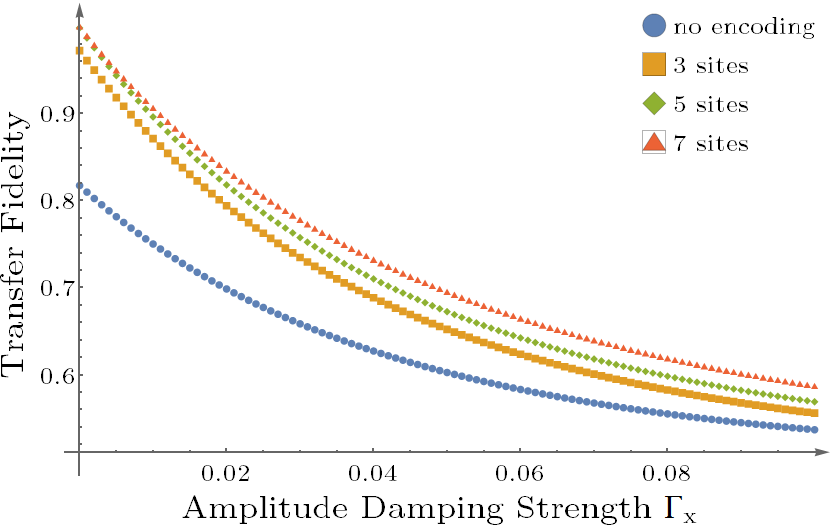}}
\subfloat[Perfect state transfer Hamiltonian \cite{christandl2004}]{\includegraphics[width=0.45\textwidth]{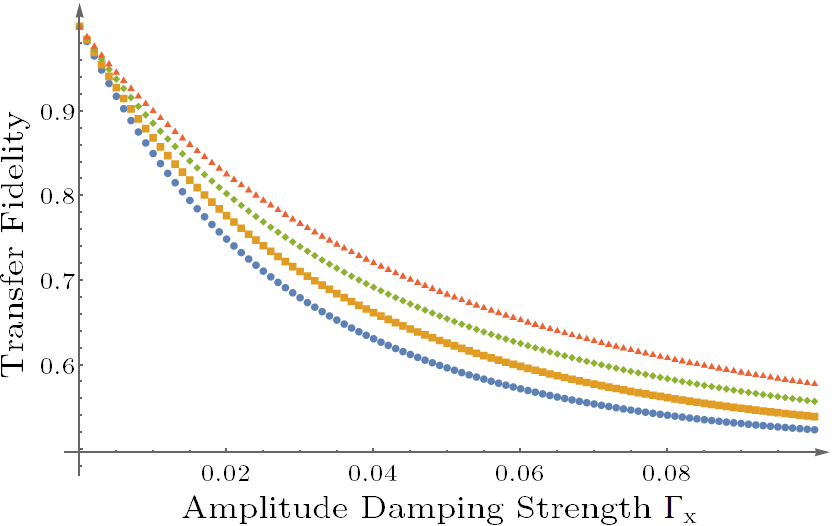}}
\caption{Performance of different Hamiltonians in the presence of amplitude damping noise, chain length 35, combatted with encoding of differing sizes.}\label{fig:ampdamp}
\end{figure*}

In this instance, we have that $S=e^{-\Gamma_xt}S_0$, where $S_0$ was the $S$ matrix in the absence of noise. Similarly,
\begin{align*}
    R &= e^{-2\Gamma_{x}t}\sum\ket{ j}\bra{ i}\bra{\textbf{mm}} e^{\mathbfcal{H}_1t}\ket{\textbf{ij}} \\
    &=e^{-2\Gamma_{x}t}\sum_{i,j}\ket{ j}\bra{i}\bra{\textbf{j}} e^{iH_1t}P_{\text{out}} e^{-iH_1t}\ket{\textbf{i}} \\
    &=S^\dagger S.
\end{align*}
Again, $R$ does not affect the choice of $\ket{\psi}$. Hence, the strength of the noise does not affect the choice of $\ket{\psi}$, so one can find the optimal encoding in the noiseless case (i.e.\ utilising Haselgrove's technique), and this is the optimal encoding for all noise strengths, just with reduced fidelity
$$
\bar F_{\Gamma_x}=\frac13+\frac16\left(1+e^{-\Gamma_xt}\left(\sqrt{6\bar F_{0}-2}-1\right)\right)^2.
$$
That said, when considering optimising over time as well, adding noise will tend to bring the optimal time (marginally) earlier.

In Fig.\ \ref{fig:ampdamp}, we plot the effects of amplitude damping noise (having optimised for time) for two different Hamiltonians. It is noteworthy that by the time we use an encoding/decoding region of size 7, the transfer fidelity has been significantly enhanced, and there is essentially no difference between the performance of the two Hamiltonians.

\subsection{Optimising over Components}

We have now seen a number of cases in which $R$ does not influence the optimal choice of $\ket{\psi}$, which is just selected to be the maximum eigenvector of $S^\dagger S$ as $R\propto\identity$ of $S^\dagger S$. In the case of general noise, we do not expect this to always hold, but we anticipate that $R$ will hold less relevance, and the $S$ term will dominate. This motivates our simplifying assumption that $\ket{\psi}$ will be close to being an eigenvector of $S^{\dagger}S$. In which case, we can approximate
\begin{equation}\label{eq:SSR}
    \overline{F}\approx\frac{1}{2}+\frac{1}{6}\bra{\psi} 2\sqrt{ S^{\dagger}S}+ R\ket{\psi}.
\end{equation}
Thus, $\ket{\psi}$ is just the maximum eigenvector of
\begin{equation*}\label{eq:Final}
    2\sqrt{S^{\dagger}S}+R.
\end{equation*}
Using the $\ket{\psi}$ in this way must represent a lower bound on the achievable fidelity. In contrast, the independent optimisation of the $R$ and $S$ terms yields an upper bound to this value. In Fig. \ref{upper}, we take the case of dephasing noise ($L_{i} = \sqrt{\frac{\Gamma_{z}}{2}}Z_{i}$), which is not expected to have $R\propto\identity$. For the maximum possible opportunity to see a discrepancy between the upper and lower bounds, we push the dephasing so strong as to render the transfer fidelities unusable. Even at this extreme, we see that the upper and lower bounds coincide, and thus anticipate that they will do so at all intermediate regimes. We thus expect this method to be essentially optimal across all relevant parameters.

\begin{figure}[t]
 \includegraphics[scale=0.22]{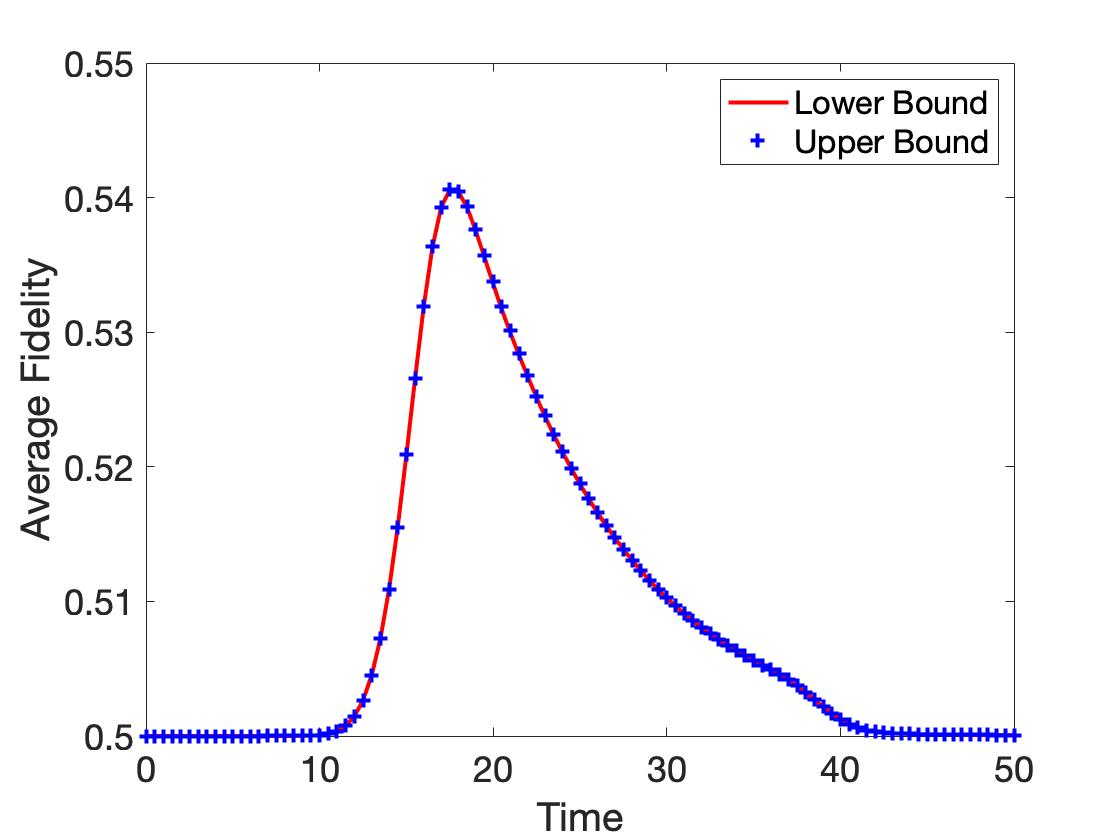}
 \caption{Comparison of upper and lower bound for average fidelity for a uniformly coupled chain of $N=35$ and dephasing strength of $\Gamma_{z} = 0.3$.}
 \label{upper}
\end{figure}

\section{Examples}

\begin{figure*}
\begin{tabular}{cc}
\subfloat[Uniform Hamiltonian with no encoding]{\includegraphics[width=0.45\linewidth]{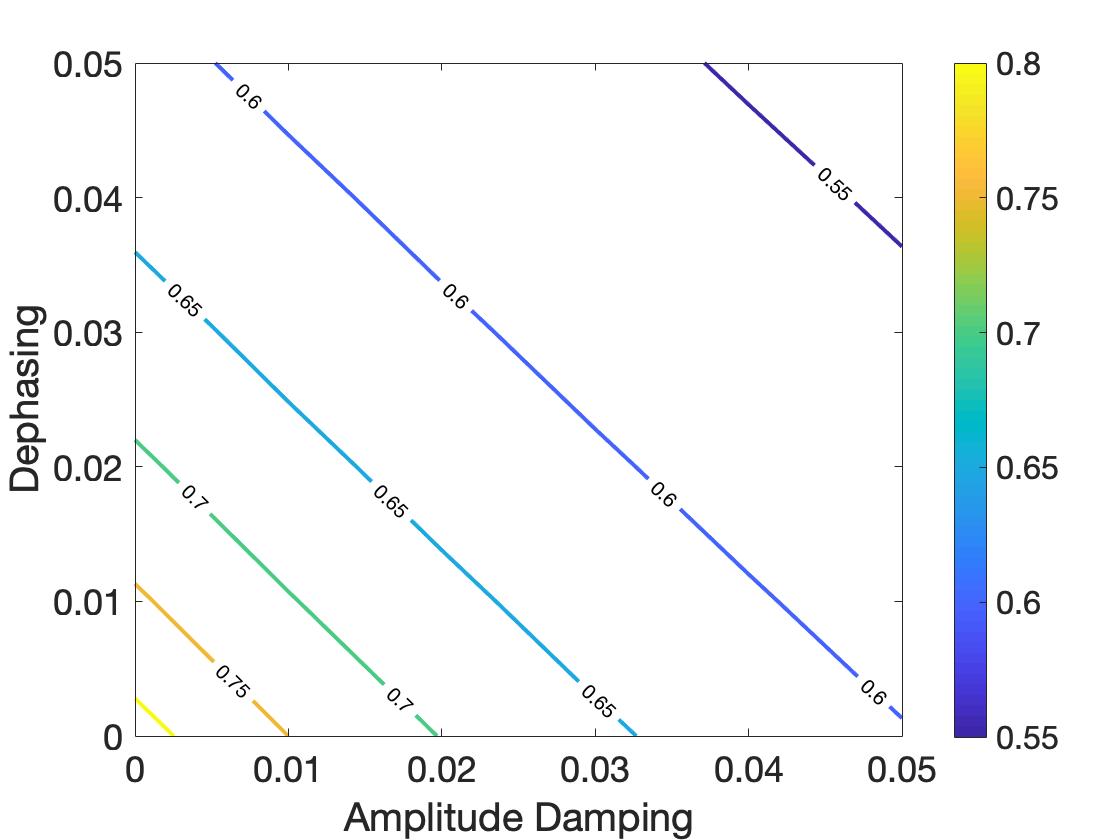}} &
\subfloat[Uniform Hamiltonian with 7-site encoding]{\includegraphics[width=0.45\linewidth]{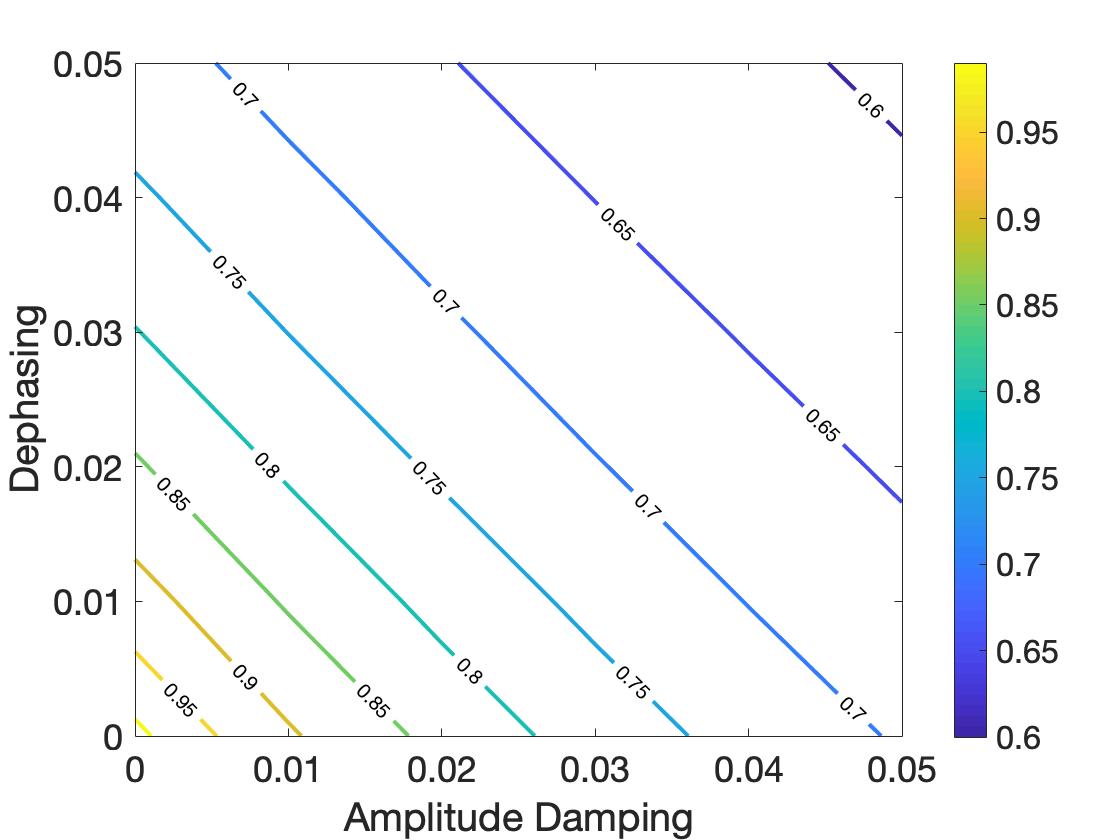}} \\
\end{tabular}
\caption{A comparison of average fidelity achieved for the uniformly coupled Hamiltonian with varying noise parameters for $N=35$.} 
        \label{fig:mean and std of nets}
    
    \end{figure*}

For the sake of concrete examples, we choose to perform transfer along the length of a chain coupled via a Hamiltonian
\begin{equation*}
    H = \frac{1}{2}\sum^{N-1}_{n=1}J_{n}(X_{n}X_{n+1}+Y_{n}Y_{n+1})-\frac12\sum^{N}_{n=1}B_{n}Z_{n}.
\end{equation*}
We impose that $J_i\leq 1$ for all $i$ in order to facilitate a fair comparison.
There are a wide variety of different coupling strengths that one could consider. We have not made a detailed study of performance of the many different options. However, the $e^{-\Gamma_x t}$ factor that arises in amplitude damping noise strongly suggests that we must find solutions that are as fast as possible. We therefore consider three cases -- the perfect state transfer chain \cite{christandl2004}, since it is the fastest such case \cite{yung2006,kay2016b}; the uniformly coupled chain \cite{bose2003}, and a tweaked version that is optimised for end-to-end transfer with high fidelity \cite{apollaro2012}.

We consider two noise models, amplitude damping, as introduced in Sec.\ \ref{sec:ampdamp}, and dephasing noise. For dephasing noise, the Lindblad operators are $L_{i} = \sqrt{\Gamma_{z}/2}Z_{i}$,
where $\Gamma_{z}$ is a measure of the strength of dephasing noise and $\Gamma_{z} = \frac{1}{T_{2}}$, where $T_{2}$ is the transverse paramagnetic relaxation time.

The results for the uniformly coupled chain are plotted in Fig.\ \ref{fig:mean and std of nets}, comparing the effects of dephasing and amplitude damping noise. These plots, and indeed all others for different Hamiltonians, consist of contours of constant fidelity that are straight lines. This is because the dominant noise term is of the form $e^{-(\Gamma_x+\Gamma_z)t}$ on a single excitation -- we already saw in Sec.\ \ref{sec:ampdamp} that this is the only contribution from amplitude damping. For dephasing noise, note that the noise terms are of the form
$$
\mathbfcal{Q}-\mathbfcal{H}=\frac{\Gamma_z}{2}\sum_nZ_n\otimes Z_n-\frac{\Gamma_z}{2}\identity\otimes\identity.
$$
This is diagonal. Of the $N^2$ diagonal elements in the single excitation subspace, $N^2-N$ of them are $-2\Gamma_z$, the exceptions being terms of the form $\ket{\textbf{n}\textbf{n}}$, which are 0. Informally, then, the dominant behaviour will be similar to $-2\Gamma_z\identity$, albeit with some correction due to the $\{\ket{\textbf{n}\textbf{n}}\}$, the most relevant component of which is that the exponential decay tends towards the maximally mixed state of the subspace, instead of leaking out of the subspace (as was the case for amplitude damping). It is this to which we ascribe the reason for a given fidelity contour being realised by a marginally higher value of $\Gamma_z$ than $\Gamma_x$.

Given that the dominant noise term is for the form $e^{-(\Gamma_x+\Gamma_z)t}$, it is clear that for a fixed transfer fidelity, we follow a contour of approximately constant $\Gamma_x+\Gamma_z$, which is what we observe. As such, it is not necessary to reproduce versions of Fig.\ \ref{fig:mean and std of nets} for multiple different encoding region sizes and Hamiltonian models. Instead, it is sufficient to refer back to Fig \ref{fig:ampdamp}, which provides the equivalent plots with $\Gamma_z=0$ fixed.
 
\section{Towards a Multiple-Excitation Encoding}\label{sec:multi}

So far, we have concentrated on encoding into the single-excitation subspace. This was partially motivated by the observation in \cite{keele2021} that the optimal encoding choice in the case of unitary evolution is in the single excitation subspace, and that this outperforms any usage of error correction. Nevertheless, it is certainly possible that using the higher excitation subspaces and error correcting codes could enhance the protection against noise -- loss of excitations from the decoding region is our source of error, so if we had several excitations and could tolerate the loss of all but one, this is going to achieve a higher fidelity (although, exactly the same expectation can be applied, erroneously, to the unitary evolution case).

In fact, the formalism developed so far is readily generalised, being aware of the discrepancy indicated in \cite{keele2021} for the case of decoding when multiple excitations are present, compared to what Haselgrove originally claimed to be optimal \cite{haselgrove2005}.

Consider a decoding region of size $|\Lambda_{\text{out}}|=M$ and let us encode using the excitation subspaces up to and including the $k^{th}$. We shall restrict to $k< M/2$. The purpose behind this assumption on the excitation number is that it lets us define $U$ much as before: $U\ket{\textbf{0}}=\ket{\textbf{0}}$ and
$$
U\ket{x}=\ket{N,\phi_x}
$$
for any $x\in\{0,1\}^M$ with weight $w_x\geq 1$ and $w_x\leq k$, where the $\ket{\phi_x}$ are orthonormal. Here, we have used $x$ to describe the basis state of the $M$ qubits in the decoding region. There are
$$
\sum_{n=1}^k\binom{M}{n}
$$
of these, and the states $\ket{\phi_x}$ only have support on $M-1$ qubits, so the maximum number must be $2^{M-1}$. Since $\sum_{n=1}^M\binom{M}{n}=2^M-1$, we need to pick $k$ to restrict to less than half the possible total sum, i.e.\ $k< M/2$. For larger weights, we cannot achieve the orthonormal condition \footnote{In \cite{keele2021}, we decoded onto a separate qubit, which removes this constraint.}, and a further approximation would be necessary. We should also note that the exact solution for amplitude damping is no longer applicable.

Repeating the previous calculations, the $\bra{11}\mathbfcal{R}\ket{\psi\psi^*}$ term becomes otherwise independent of the choice of $U$, and we recover a similar matrix to before,
\begin{equation}\label{eqn:Rmulti}
R=\sum_{x,y\in\{0,1\}^{|\Lambda_{\text{in}}|}}\sum_{\substack{z\in\{0,1\}^N\\w_z\leq k\\w_z^{\text{out}}\geq 1}}\ket{y}\bra{x}\bra{z,z}e^{\mathbfcal{Q}t}\ket{x\textbf0,y\textbf0}.
\end{equation}
Here, $w_z^{\text{out}}$ indicates the weight of the bit string restricted only to the components in $\Lambda_{\text{out}}$.
Similarly, optimisation of the $\bra{10}\mathbfcal{R}\ket{\psi\textbf{0}}$ term requires selection of the unitary such that
$$
\braket{\textbf{0}}{\phi_x}=\frac{\gamma_{x0,00}^*}{\sqrt{\displaystyle\sum_{z\in\{0,1\}^M}\left|\gamma_{z0,00}\right|^2}}
$$
where
$$
\gamma_{x,y}=\bra{x,y}e^{\mathbfcal{Q}t}\ket{\psi\textbf0}.
$$
This leads to a matrix
\begin{align*}
S&=\sum_{\substack{x\in\{0,1\}^M\\w_x\geq 1}}\ket{x}\bra{x\textbf{0},\textbf{00}}e^{\mathbfcal{Q}t}\sum_{\substack{z\in\{0,1\}^{|\Lambda_{\text{in}}|}\\w_z\geq 1}}\ket{z\textbf{0},\textbf{00}}\bra{z}. \\
&=(P_{\text{out}}\otimes\bra{\textbf{0}})e^{\mathbfcal{Q}t}(P_{\text{in}}\otimes\ket{\textbf{0}}).
\end{align*}
This time $P_{\text{in}}$ and $P_{\text{out}}$ are projectors onto the 1 to $k$ excitation subspaces on the input and output regions respectively, and onto all other qubits being in the $\ket{0}$ state.

For all the cases we have been considering, $S$ divides into a block-diagonal structure based on excitation number. For excitation preserving $Q$, this is trivial -- $\ket{z_\text{in}\textbf{0}}\ket{\textbf{0}}\mapsto \ket{x_{\text{out}}\textbf{0}}\ket{\textbf{0}}$ only if $x$ and $z$ have the same weight. For amplitude damping noise, the excitation-decreasing terms are due to terms of the form $(X+iY)\otimes(X+iY)$. This requires that we are able to remove an excitation from both halves of $\ket{\psi}\ket{\textbf{0}}$, which is clearly not possible -- they have no effect, and we revert to the excitation preserving case.

$R$ has an identical subspace structure. To see this, consider the term
$$
\bra{zz}e^{\mathbfcal{Q}t}\ket{x\textbf{0},y\textbf{0}}
$$
from Eq.\ (\ref{eqn:Rmulti}). From our previous discussion, we know that the Hamiltonian and noise will either preserve the number of excitations of $x$ and $y$, or decrease them by an equal number. But since they must both end up having the same excitation number ($w_z$), they must have started with the same excitation number. Note, however, that although both $S$ and $R$ are block diagonal, it does not necessarily mean that $\ket{\psi}$ is always supported on just one excitation subspace -- this is an effect of the square root in Eq.\ (\ref{eqn:fbar}).

\begin{figure}
\includegraphics[width=0.45\textwidth]{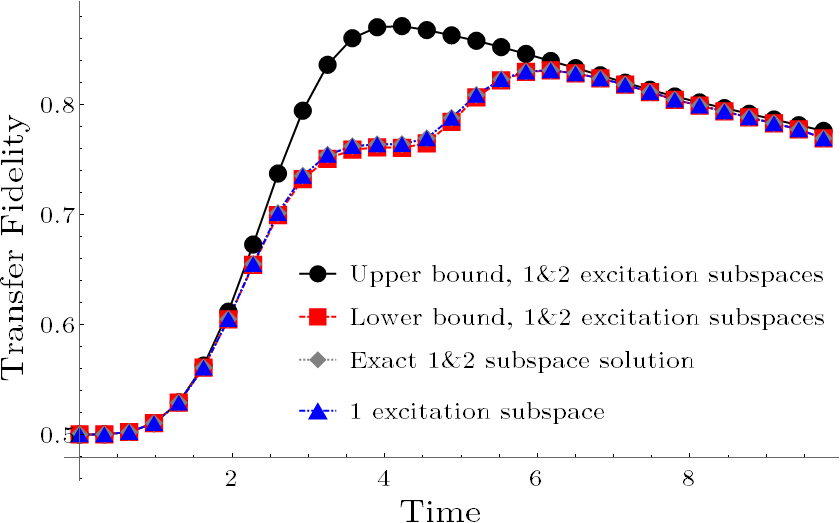}
        \caption{Average fidelity over time for a uniformly coupled chain of $N=13$. Encoding/decoding region of 5 qubits. We compare the upper bound for fidelity in the two excitation case with a lower bound and an exact solution. The exact solution case is hidden under the single excitation case. There is no advantage to using higher excitation subspaces. ($\gamma_x=0.02,\gamma_z=0.04$)} 
        \label{fig:TimeHigherSubspaces}
    \end{figure}

With both $S$ and $R$ in place, Eq.\ (\ref{eqn:fbar}) still holds. Does the approximation in Eq.\ (\ref{eq:SSR}) remain useful? Does encoding in higher excitation subspaces yield an improved fidelity?

In the noise-free case, so long as the maximum transfer amplitude in the single excitation subspace is at least $\sqrt{2}-1$, the optimal encoding is in the single excitation subspace. By continuity, we expect this to remain true for weak noise as well.

We have not performed extensive studies of the multiple excitation subspaces, which are computationally far more demanding. However, we have performed simple tests on more modestly sized systems, and have never found an improvement by going to higher excitation subspaces. In Fig.\ \ref{fig:TimeHigherSubspaces}, we depict a typical case; a single noisy instance of a uniformly coupled chain. For reference, we display (blue triangles), the behaviour in the single excitation subspace. In solid black circles, we show an upper bound on the achievable fidelity based on separate optimisation of $\ket{\psi}$ for the matrices $R$ and $S$. In red squares, we show the fidelity achieved due to the approximation of Eq.\ (\ref{eq:SSR}), a lower bound of what can be achieved. In higher excitation subspaces, the coincidence of these two lines is clearly not as tight as it was in the first excitation subspace. For such small cases, we can exactly find the optimal solution (grey diamonds). In all such cases we have tried, it has always been close to the lower bound that is the best achievable value. 



    

\section{Conclusion}\label{Conclusion}

We have a presented a simple to implement scheme that can be applied to a wide range of Hamiltonians to improve fidelity in the presence of noise, by giving near-optimal encoding strategies. This is applicable to any Hamiltonian that is excitation-preserving. We have demonstrated that larger encoding and decoding regions lead to better transfer in the presence of noise although reducing the overall transfer distance. Even modest sizes of encoding region can convert scenarios that are impossible for transfer into reasonable propositions. This is most compelling in the case of the uniform chain, which suggests there is little value in considering other state transfer systems with more complex, harder to implement, coupling schemes.

We have also shown that it is of benefit to choose a Hamiltonian that allows faster transfer over one that (in the absence of noise) produces higher fidelity transfer. Our scheme leads to further improvement in these faster transfer chains. This technique can be applied to NISQ devices to allow some improvement in state transfer fidelity with respect to noise without implementing a full error correction scheme. 

We have explicitly considered a specific form of Hamiltonian based on a spin chain. However, our derivation only assumed an excitation-preserving Hamiltonian, and does not depend on any underlying coupling geometry. Similarly, a broad class of noise models can be handled. We have primarily focused on the single excitation subspace, but have provided a route via which the formalism can continue to higher excitation encodings, although numerics have failed to find any gain from doing so.

Here, we focused on state transfer for the sake of having a concrete task to talk about. However, the formalism could easily be adapted to other Hamiltonian-based tasks \cite{kay2017c,kay2017a}, particularly those whose success is measured by fidelity. One might also be able to extend this work to cover other Hamiltonian models such as those that have different subspace structures \cite{kay2007,difranco2008} including those from the Jordan-Wigner transformation.

\bibliography{../../../References}
\end{document}